\documentclass[letterpaper, 10 pt, conference]{ieeeconf}

\IEEEoverridecommandlockouts
\usepackage{amsmath, amssymb, mathtools, mathrsfs}
\usepackage{mathdots}
\usepackage{cancel}
\usepackage[normalem]{ulem}
\usepackage{nicematrix}
\usepackage{makecell}
\usepackage{tikz}
\usepackage{nicematrix}

\usepackage{eso-pic}

\let\labelindent\relax
\usepackage{enumitem}

\usepackage{amsthm}

\newtheorem{theorem}{Theorem}
\newtheorem{lemma}{Lemma}
\newtheorem{proposition}{Proposition}
\newtheorem{corollary}{Corollary}
\newtheorem{assumption}{Assumption}

\newtheorem{remark}{Remark}

\DeclareMathOperator{\im}{im}
\DeclareMathOperator{\rank}{rank}
\DeclareMathOperator{\spanop}{span}
\DeclareMathOperator{\PE}{PE}

\newcommand{\Obsvm}[2]{\mathscr{O}_{#1}\!\left(#2\right)}
\newcommand{\Ctrbm}[2]{\mathscr{C}_{#1}\!\left(#2\right)}
\newcommand{\Signal}[2]{#1_{[#2]}}
\newcommand{\Lift}[3]{#1^{\langle#2\rangle}(#3)}
\newcommand{\LiftCT}[3]{#1^{\langle#2\rangle,\,{c}}(#3)}
\newcommand{\Hankelm}[2]{H_{#1}(#2)}
\newcommand{\floor}[1]{\left\lfloor #1 \right\rfloor}
\newcommand{\R}{\mathbb{R}}
\newcommand{\N}{\mathbb{N}}
\newcommand{\Z}{\mathbb{Z}}

\NiceMatrixOptions
{
    custom-line ={command= H, tikz= dashed, width= 1mm},
    custom-line = {letter= I, tikz= dashed, width= 1mm},
}
\NiceMatrixOptions{cell-space-limits=1pt}

\title{\LARGE \bf A note on input signal generators:\\
A relaxation of Willems' fundamental lemma in the SISO case}

\author{Yun Jeong Yang and Jin Gyu Lee
\thanks{This work was supported in part by the National Research Foundation of Korea
(NRF) grant funded by the Korea government (MSIT) (RS-2026-25504174) and the Strategic Hub for International Research Collaboration project of Seoul National University (SNU).}
\thanks{Y. J.~Yang and J. G.~Lee are with ASRI and the Department of Electrical and Computer Engineering, 
        Seoul National University, Korea (e-mail: {\tt\small yjyang@cdsl.kr} and {\tt\small jingyu.lee@snu.ac.kr}). Corresponding author: Jin Gyu Lee.}
}

\begin{document}

\maketitle
\thispagestyle{empty}
\pagestyle{empty}

\AddToShipoutPictureFG*{%
  \AtPageLowerLeft{%
    \put(42,35){%
      \parbox{0.86\paperwidth}{%
        \scriptsize
        \copyright~2026 IEEE. Personal use of this material is permitted.  Permission from IEEE must be obtained for all other uses, in any current or future media, including reprinting/republishing this material for advertising or promotional purposes, creating new collective works, for resale or redistribution to servers or lists, or reuse of any copyrighted component of this work in other works.
      }%
    }%
  }%
}

\begin{abstract}
We provide a practical relaxation of Willems' fundamental lemma for discrete-time linear time-invariant (single-input-single-output) systems. 
Instead of maintaining conventional Willems' persistency of excitation condition in the behavioral theory, we reformulate the problem in terms of signal generators, hence going back to the dynamical systems theory.
We discuss the relationship between the persistency of excitation order and the dimension of the signal generator.
Furthermore, we identify a necessary and sufficient condition on the signal generator that can generate informative input--output data for almost all systems and initial conditions.
This formulation accommodates both a broad range of conventional persistently exciting input designs and a class of shorter inputs that violate Willems’ condition, including sinusoidal sequences with fewer frequencies.
Finally, the signal generator perspective allows a natural extension to continuous-time systems.
\end{abstract}

{\section{Introduction}}
Various methodologies for identifying or controlling unknown systems must rely on input--output (and perhaps state) data.
In particular, for discrete-time linear time-invariant (LTI) systems, an essential lemma known as Willems' fundamental lemma states that 
suitable inputs can generate input--output data that can completely characterize the system's input--output behavior {\cite{i1}}.
This lemma laid the foundation of most of the data-driven control literature such as {\cite{ia1,ia2}}.

The main role of Willems' fundamental lemma can be divided into two aspects.
First, it determines when the collected data be sufficient to recover the system behavior.
We refer to such data---the finite number of input--output windows that spans the space of all possible input--output windows---as informative input--output data.
Subsequent studies have further clarified and refined the data conditions for entire system recovery (and other objectives) {\cite{ib1,ib2}}.
Second of all, Willems' fundamental lemma provides a guideline dedicated to designing inputs to collect informative input-output data.

According to Willems' fundamental lemma, we can obtain informative input--output data if the input signal is persistently exciting of sufficiently high order.
Since this condition is not a necessary condition, subsequent studies have investigated whether informative input--output data can be obtained under weaker conditions on inputs.
In particular, {\cite{ic1,ic2}} have shown that informative data can be collected using shorter input sequences by actively designing future inputs based on past input--output samples, which is known as an online method.
On the other hand, {\cite{i2,i3}} have shown that if Willems' condition is not satisfied, then some system fails to produce informative input--output data for at least one initial condition.
However, these results do not clarify whether such failures are generic or confined to exceptional initial conditions.
This motivates the question of whether, under offline input design, one can characterize a broader class of input sequences that violate Willems' condition 
while still producing informative input--output data for \emph{almost all} systems and initial conditions.

The Willems' condition, written in terms of persistency of excitation, characterizes an input through the rank of the Hankel matrix, but it does not explicitly describe the dynamics underlying an input generation. 
The condition is also formulated only for discrete-time inputs. 
To obtain a richer description of input sequences beyond the Willems' condition and thereby enable a finer analysis of input--output data, 
we focus instead on the dynamics that generates input. 
Specifically, we consider the autonomous LTI system
\begin{equation*}
    \begin{aligned}
        w(t+1) &= S_gw(t),\\
        u(t)&= L_gw(t),
    \end{aligned}
\end{equation*}
which we call a signal generator.
As will be shown in Section~III, 
this signal generator formulation covers a broad range of conventional persistently exciting inputs widely used in practice, 
including all multisine signals and almost all randomly generated signals.
It also allows for natural extension to the continuous-time counterpart.
Most importantly, it provides an explicit form of underlying dynamics of an input sequence.

In this respect, we propose to go back from behavioral theory to dynamical systems theory.
Shifting away from persistency of excitation, we direct our attention to the input signal generator and characterize the input conditions that make the input--output data informative.
This transition does more than merely restating Willems' fundamental lemma; 
it yields a practically meaningful relaxation.
For example, to recover the input--output behavior of a system with known order \(n\), a direct application of Willems' condition requires a signal generator of dimension at least \(2n+1\). 
In contrast, our result shows that dimension \(n+1\) is sufficient for almost all systems and initial conditions.

Meanwhile, this type of signal generator plays a crucial role in moment matching based model reduction theory {\cite{i4}}. 
It is known that if we generate inputs using a signal generator with sufficiently high dimension, the behavior of the system can be observed across a rich frequency band {\cite{i5}}.
This aligns well with the philosophy of Willems' fundamental lemma that a persistently excited input must be applied.

Taken together, this paper makes the following three contributions based on the signal generator interpretation of input generation:
\begin{enumerate}[label=(\alph*)]
    \item We establish a relationship between the persistency of excitation condition and the signal generator representation of inputs (Section~III).
    
    \item We characterize a necessary and sufficient condition under which a signal generator produces informative input--output data for \emph{almost all} systems and initial conditions (Section~IV).

    \item We derive a natural extension of the fundamental lemma to continuous-time LTI systems with a continuous-time counterpart of the signal generator (Section~V).
\end{enumerate}
The contribution (b) is particularly meaningful since it identifies a category of useful inputs that were previously overlooked due to its failure to meet Willems' condition.
For example, shorter multisine sequences with fewer frequencies can be used as inputs to collect informative input--output data.

In Section~II, we briefly recall the standard notions on data-driven analysis.
In Section~III, we discuss the relationship between the persistency of excitation criteria and the dimension of the signal generator, and revisit Willems' fundamental lemma using signal generator based input conditions. 
Section~IV presents our main result, namely a new fundamental lemma, and demonstrates how Willems' persistency of excitation condition can be relaxed.
In Section~V, we naturally extend our result to continuous-time systems. 
Concluding remarks are presented in Section~VI.

\textit{Notation:} 
\(\N\) and \(\Z\) denote the set of natural numbers and integers, respectively. 
\(\R\), \(\R^n\), and \(\R^{n\times m}\) denote the sets of real numbers, \(n\)-dimensional real vectors, and \(n\)-by-\(m\) real matrices, respectively.
\(I_n\) denotes the \(n\)-by-\(n\) identity matrix and \(0_{n\times m}\) denotes the \(n\times m\) zero matrix.
We write \(I\) and \(0\) when the dimensions are clear from the context. Blank entries in displayed matrices are understood to be zero.
\(\sigma(M)\) denotes the spectrum of a square matrix \(M\).
For a matrix \(M\), \(\rank(M)\), \(\im(M)\), and \(\ker(M)\) denote its rank, image, and null space, respectively.
\(\dim(V)\) denotes the dimension of a vector space \(V\) and \(\spanop(S)\) denotes the linear span of a set \(S\).
\(\emptyset\) denotes the empty set. 
For matrices \(A \in \R^{n\times n}\), \(B \in \R^{n\times 1}\), and \(C \in \R^{1\times n}\),
\begin{equation*}
    \begin{aligned}
        \Ctrbm{k}{A,B}
        &\coloneqq\! \begin{bmatrix}
        B & AB & \cdots & A^{k-1}B
        \end{bmatrix},\\
        \Obsvm{k}{C,A}
        &\coloneqq\! \Ctrbm{k}{A^\top,C^\top}^\top
    \end{aligned}
\end{equation*}
denote the \(k\)-step controllability and observability matrices, respectively.
If \(k\) is omitted, it is understood \(k=n\).
For a signal \(z:\Z\to\R\), \(t_0\in\Z\), and \(T\in\N\), \(\Signal{z}{t_0,t_0+T-1}\) denotes the finite sequence \(z(t_0),z(t_0+1),\ldots,z(t_0+T-1)\).

{\section{Preliminaries}}
We address the problem of identifying a finite-dimensional LTI system from an input--output data.
Since such data determines only the input--output map, only the controllable and observable part of the system can be recovered.
Hence, for clarity, we restrict our attention to unknown systems that are controllable and observable.
Moreover, we confine our discussion to the single-input-single-output (SISO) case for clear exposition.

We first consider the discrete-time setting and then extend the discussion to the continuous-time setting.
Consider a discrete-time LTI SISO system
\begin{equation}\label{eq:sys}
    \begin{aligned}
        x(t+1) &= Ax(t)+Bu(t)\in\R^n,\\
        y(t)   &= Cx(t)+Du(t),
    \end{aligned}
\end{equation}
which is controllable and observable.

{\subsection{Data representation of systems}
For a signal \(z:\Z\to\R\) and \(L\in\N\), let
\[
    \Lift{z}{L}{t}\coloneqq [z(t)\; z(t+1)\; \cdots\; z(t+L-1)]^\top \in \R^L.
\]
We call \(\Lift{z}{L}{t}\) the \(L\)-window of \(z\) starting at \(t\in\Z\).
Then, we define the set of input--output \(L\)-windows of the system~\eqref{eq:sys} as
\[
    \mathscr{B}_L
    \coloneqq
    \left\{\!
    \begin{bmatrix}
        \Lift{u}{L}{t_0}\\[1mm]
        \Lift{y}{L}{t_0}
    \end{bmatrix}\!\in\R^{2L}
    \;\middle|
    \begin{array}{l}
        \exists\, x:\Z\to\R^n \text{ such that}\\
        (u,x,y)\text{ satisfies } \eqref{eq:sys}\\
        \text{for } t\in\{t_0,\ldots,t_0+L-1\}
    \end{array}
    \!\!\!\right\}
\]
and call each element an input--output \(L\)-window of~\eqref{eq:sys}.
Since~\eqref{eq:sys} is an LTI system, \(\mathscr{B}_L\) is a linear subspace of \(\R^{2L}\).

We consider the problem of recovering \(\mathscr{B}_L\) from input--output data, which includes system identification as its special case.
Since \(\mathscr{B}_L\) is a linear subspace, we aim to find a spanning set of \(\mathscr{B}_L\) from data.
For \(T\geq L\), let \(\Signal{u}{0,T-1}\) and \(\Signal{y}{0,T-1}\) be an input--output data generated by the system~\eqref{eq:sys}.
Then, we can obtain \(T-L+1\) input--output \(L\)-windows 
\[
    \begin{bmatrix}
        \Lift{u}{L}{0}\\
        \Lift{y}{L}{0}
    \end{bmatrix},
    \begin{bmatrix}
        \Lift{u}{L}{1}\\
        \Lift{y}{L}{1}
    \end{bmatrix},\ldots,
    \begin{bmatrix}
        \Lift{u}{L}{T-L}\\
        \Lift{y}{L}{T-L}
    \end{bmatrix},
\]
from this input--output data. 
We further define the input--output Hankel matrix of depth \(L\) by
\[\setlength{\arraycolsep}{2.5pt}
    \begin{aligned}
        \mathscr{H}_{L,T}(u,y)
        \coloneqq&\! \begin{bmatrix}
            \Lift{u}{L}{0}& \cdots& \Lift{u}{L}{T-L}\\
            \Lift{y}{L}{0}& \cdots &\Lift{y}{L}{T-L}
        \end{bmatrix}\!\in\R^{2L\times (T-L+1)}.
    \end{aligned}
\]
Then, \(\im(\mathscr{H}_{L,T}(u,y))\) represents the linear span of 
the collected windows. It is clear that \(\im(\mathscr{H}_{L,T}(u,y))\) is a linear subspace of \(\mathscr{B}_L\).

Therefore, if we have an input--output data that satisfies
\begin{equation}\label{eq:imHankel}
    \im(\mathscr{H}_{L,T}(u,y))=\mathscr{B}_L,
\end{equation}
that is, the set of collected input--output \(L\)-windows spans the entire \(\mathscr{B}_L\), then we can recover all possible input--output \(L\)-windows.
The following is the well-known equivalent rank characterization {\cite{m1}}.

\begin{proposition}[{\cite[Proof of Theorem~6]{m1}}]\label{prop:colspace}
Consider an input-output data \((\Signal{u}{0,T-1},\Signal{y}{0,T-1})\) generated by~\eqref{eq:sys}. 
For \(L\leq T\), the following are equivalent:
\begin{enumerate}[label=(\alph*)]
    \item The condition~\eqref{eq:imHankel} is satisfied.
    \item \(\rank(\mathscr{H}_{L,T}(u,y))
    \!=\! \dim(\mathscr{B}_L)
    \!=\!L + \rank(\Obsvm{L}{C,\mkern-1muA})\).
\end{enumerate}
\end{proposition}
Throughout the paper, we call an input--output data \((\Signal{u}{0,T-1},\Signal{y}{0,T-1})\) informative for \(L\) if it satisfies either of the equivalent conditions in Proposition~\ref{prop:colspace}.\footnote{
    More generally, informative input--output data can be defined by spanning sets of \(L\)-windows of \(\mathscr{B}_L\). Here, we restrict our attention to the case where such a spanning set is generated from a single consecutive input--output sequence.
} 
When the length \(L\) is clear from the context or not essential to the discussion, we simply say that the input--output data is informative.

\subsection{Brief summary of Willems' fundamental lemma}
Willems' fundamental lemma suggests a sufficient condition on inputs that yields an informative input--output data. 
This condition is called persistency of excitation in~{\cite{i1}}.
In particular, \(\Signal{u}{0,T-1}\) is persistently exciting (PE) of order \(K\) if 
\[
    \begin{aligned}
        \Hankelm{K}{\Signal{u}{0,T-1}}\coloneqq& \begin{bmatrix}
            u(0)&u(1)&\cdots & u(T-K)\\
            u(1)&u(2)&\cdots & u(T-K+1)\\
            \vdots&\vdots&\ddots&\vdots\\
            u(K-1)&u(K)&\cdots & u(T-1)
        \end{bmatrix}\\
        =& \begin{bmatrix}
            \Lift{u}{K}{0}&\Lift{u}{K}{1}&\cdots & \Lift{u}{K}{T-K}
        \end{bmatrix}
    \end{aligned}
\] 
has full row rank.
We denote the maximum PE order of \(\Signal{u}{0,T-1}\) by \(\PE(\Signal{u}{0,T-1})\).\footnote{
    If \(\Signal{u}{0,T-1}\) is PE of order \(K\), then it is PE of order \(k\) for all \(k\leq K\); if it is not PE of order \(K'\), then it is not PE of order \(k\) for all \(k\geq K'\). Thus, maximum PE order is well-defined.
}
It straightforwardly satisfies \(0 \leq \PE(\Signal{u}{0,T-1}) \leq \floor{\tfrac{T+1}{2}}\).
With this notation, Willems' fundamental lemma can be stated as follows.
\begin{proposition}[{\cite{i1,i2,i3}}]\label{prop:WFL+}
Suppose that we apply \(\Signal{u}{0,T-1}\) to the system~\eqref{eq:sys} as an input.\footnote{
    The original Willems' fundamental lemma in {\cite{i1}} only considered the case (c). For the remaining parts, (a) is from \cite{i2} and (b) is from \cite{i3}.
}
\begin{enumerate}[label=(\alph*)]
    \item If \(\PE{(\Signal{u}{0,T-1})}<L\), then the input--output data fails to be informative for \(L\), for \emph{any} \(x(0)\in\R^n\).
    \item If \(L\leq \PE{(\Signal{u}{0,T-1})}< L+n\), then the input--output data fails to be informative for \(L\), for \emph{some} \(x(0)\in\R^n\), generically.\footnote{
        Here, \emph{generically} means that, for a given input \(\Signal{u}{0,T-1}\), the statement holds for almost all controllable systems. 
        See~\cite{i3} for a precise statement and the associated conditions.
    }
    \item If \(\PE{(\Signal{u}{0,T-1})}\geq L+n\), which in particular requires \(\quad T\geq 2\PE(\Signal{u}{0,T-1})-1\), then the input--output data is informative for \(L\), for \emph{every} \(x(0)\in\R^n\).
\end{enumerate}
\end{proposition}
Building upon this summary, we provide further interpretations of these classical results from the perspective of a signal generator in the next section.

{\section{Inputs generated by signal generators}
The following autonomous LTI single-output system
\begin{equation}\label{eq:sgenDT}
    \begin{aligned}
        w(t+1)&=S_gw(t)\in\R^{N_g},\quad w(0)=w_0\\
        u(t)&=L_gw(t)\in\R,
    \end{aligned}
\end{equation}
provides a natural model for an input signal generator.
We use the response of the signal generator~\eqref{eq:sgenDT} initialized at \(w_0\)
as an input to the system~\eqref{eq:sys}.

We want to consider a minimal signal generator that generates the given input sequence.
So, we assume the following.
\begin{assumption}\label{ass:sglg}
The pair \((L_g,S_g)\) is observable.
\end{assumption}
\begin{assumption}\label{ass:sgw0}
The pair \((S_g,w_0)\) is controllable.
\end{assumption}

Since a signal generator is a design object, these assumptions are well justified.
In fact, given any signal generator that does not satisfy these assumptions, 
the Kalman decomposition gives a lower-order signal generator satisfying both assumptions while generating the same signal.

\subsection{Relationship between the maximum PE order and the dimension of the signal generator}
The following lemma characterizes the relationship between the state dimension \(N_g\) of the signal generator and the maximum PE order of a given signal.

\begin{lemma}\label{lem:sgenPE} The following are equivalent.
\begin{enumerate}[label=(\alph*)]
    \item \(\Signal{u}{0,T-1}\) is a response of the signal generator~\eqref{eq:sgenDT} satisfying Assumptions~\ref{ass:sglg}--\ref{ass:sgw0}, with \(N_g=K\) and \(T\geq 2K-1\).
    \item \(\Signal{u}{0,T-1}\) is a signal that satisfies \(\PE(\Signal{u}{0,T-1})=K\) and 
    \(\rank(\Hankelm{K}{\Signal{u}{0,T-2}}) = \rank(\Hankelm{K+1}{\Signal{u}{0,T-1}})\).
\end{enumerate}
\end{lemma}
\begin{proof}

((a) \(\Rightarrow\) (b)) It was established in {\cite[Lemma~1]{i2}}. 
For completeness, we include a self-contained proof.

We have \(u(t) = L_gS_g^{t}w_0\) for all \(t\ge 0\). Thus,
\begin{equation*}
    \begin{aligned}
        \Hankelm{K}{\Signal{u}{0,T-2}} &= \Obsvm{K}{L_g, S_g}\Ctrbm{T-K}{S_g,w_0},\\
        \Hankelm{K}{\Signal{u}{0,T-1}} &= \Obsvm{K}{L_g, S_g}\Ctrbm{T-K+1}{S_g,w_0},\\
        \Hankelm{K+1}{\Signal{u}{0,T-1}} &= \Obsvm{K+1}{L_g, S_g}\Ctrbm{T-K}{S_g,w_0}.
    \end{aligned}
\end{equation*}
Under Assumptions~\ref{ass:sglg}--\ref{ass:sgw0}, we have
\begin{equation*}
    \begin{aligned}
        &\rank(\Obsvm{K}{L_g,S_g})=\rank(\Obsvm{K+1}{L_g,S_g}) = K,\\
        &\rank(\Ctrbm{T-K}{S_g,w_0}) \leq\rank(\Ctrbm{T-K+1}{S_g,w_0}) = K.
    \end{aligned}
\end{equation*}
Hence, 
\(\rank(\Hankelm{K}{\Signal{u}{0,T-2}}) = \rank(\Hankelm{K+1}{\Signal{u}{0,T-1}})<K+1\) and \(\rank(\Hankelm{K}{\Signal{u}{0,T-1}}) = K\).

((b) \(\Rightarrow\) (a)) The rank condition 
\(\rank(\Hankelm{K}{\Signal{u}{0,T-2}}) = \rank(\Hankelm{K+1}{\Signal{u}{0,T-1}})\)
implies the existence of \(\xi = [\xi_0\; \cdots\; \xi_{K-1}]\!^\top\!\in\!\R^K\) such that
\[[u(K)\; u(K+1)\; \cdots\; u(T-1)] = -\xi^\top \Hankelm{K}{\Signal{u}{0,T-2}}.\]
Let \(S_g\) be any non-derogatory matrix\footnote{
    A non-derogatory matrix is a matrix whose minimal polynomial coincides with its characteristic polynomial. 
    There exists a row vector \(C\) such that the pair \((C,A)\) is observable if and only if \(A\) is non-derogatory.} 
with \(\det(zI-S_g)=z^{K}+\xi_{K-1}z^{K-1}+\cdots+\xi_0\),
\(L_g\) be such that \((L_g,S_g)\) is observable,
and \(w_0={\Obsvm{K}{L_g,S_g} }^{-1}\Signal{u}{0,K-1}\).
Let the response of the signal generator~\eqref{eq:sgenDT} with \(L_g,S_g,\) and \(w_0\) be \(\bar u(i)=L_gS_g^iw_0\) for \(i=0,\ldots,T-1\).

We show that \(\Signal{u}{0,T-1}=\Signal{\bar u}{0,T-1}\), inductively.
It is trivial that \(\Signal{u}{0,K-1}=\Signal{\bar u}{0,K-1}\).
Assume that \(\Signal{u}{0,i-1}=\Signal{\bar u}{0,i-1}\) for some \(K\leq i\leq T-1\).
Then, by Cayley-Hamilton theorem,
    \begin{align*}
        \bar u(i)&=L_gS_g^iw_0=L_g(-\xi_{K-1}S_g^{i-1}-\cdots-\xi_{0}S_g^{i-K})w_0\\
        &=-\xi_{K-1}\bar u(i-1)-\cdots-\xi_{0}\bar u(i-K)\\
        &=-\xi^\top \Lift{u}{K}{i-K}=u(i).
    \end{align*}

Finally, \((S_g, w_0)\) is controllable since
\(\Ctrbm{T-K+1}{S_g,w_0}= \Obsvm{K}{L_g,S_g}^{-1}\Hankelm{K}{\Signal{u}{0,T-1}}\) has full row rank.
\end{proof}

\subsection{Sinusoidal and random signals can almost always be generated by a signal generator}

We now examine the class of inputs that, among those satisfying \(\PE(\Signal{u}{0,T-1})=K\), 
can be generated by the signal generator~\eqref{eq:sgenDT} satisfying Assumptions~\ref{ass:sglg}--\ref{ass:sgw0}, with state dimension \(N_g=K\).
Equivalently, in view of Lemma~\ref{lem:sgenPE}, we investigate how restrictive the additional condition 
\(\rank\left(\Hankelm{K}{\Signal{u}{0,T-2}}\right) = \rank\left(\Hankelm{K+1}{\Signal{u}{0,T-1}}\right)\)
in Lemma~\ref{lem:sgenPE} (b) is.
We show that this class is sufficiently broad to encompass a substantial range of conventional persistently exciting inputs.

A signal with maximum PE order \(K\) was typically generated by two methods.
The first is to generate a multisine input having \(K/2\) distinct frequencies if \(K\) is even, and \((K-1)/2\) distinct frequencies together with a constant bias if \(K\) is odd.
They can always be generated by the signal generator~\eqref{eq:sgenDT} of dimension \(K\).

The second is to generate a random signal.
The following lemma guarantees that almost all randomly generated input sequences of length \(2K-1\) or \(2K\) 
not only have maximum PE order \(K\) but can also be generated by the signal generator~\eqref{eq:sgenDT} of dimension \(K\).

\begin{lemma}\label{lem:sgenPErandom}
Suppose that the terms of \(\Signal{u}{0,T-1}\) are drawn independently from the uniform distribution on \([a,b]\subset \R\), where \(a<b\).
Then the following hold with probability \(1\).
\begin{enumerate}[label=(\alph*)]
    \item \(\PE(\Signal{u}{0,T-1})=\floor{\tfrac{T+1}{2}}\).
    \item \(\Signal{u}{0,T-1}\) is a response of the signal generator~\eqref{eq:sgenDT}, satisfying Assumptions~\ref{ass:sglg}--\ref{ass:sgw0}, with \(N_g=\floor{\tfrac{T+1}{2}}\).
\end{enumerate}
\end{lemma}
\begin{proof}
Let \(K = \floor{\tfrac{T+1}{2}}\).
For even \(T\), \(\rank(\Hankelm{K}{\Signal{u}{0,T-2}})=K\) almost surely by Lemma~\ref{lem:randomPEunif} in Appendix.
It follows that \(\rank(\Hankelm{K}{\Signal{u}{0,T-1}})=\rank(\Hankelm{K+1}{\Signal{u}{0,T-1}})=K\). 
For odd \(T\), \(\rank(\Hankelm{K}{\Signal{u}{0,T-1}})=K\) almost surely by Lemma~\ref{lem:randomPEunif} in Appendix. 
It follows that \(\rank(\Hankelm{K}{\Signal{u}{0,T-2}}) = \rank(\Hankelm{K+1}{\Signal{u}{0,T-1}})=K-1\).
So, (a) and (b) both hold for almost all \(\Signal{u}{0,T-1}\) by Lemma~\ref{lem:sgenPE}.
\end{proof}

This subsection shows that the signal generator formulation encompasses an extensive range of persistently exciting inputs encountered in practice. 
It therefore retains much of the practical scope of the conventional persistency of excitation framework while providing structural information beyond that conveyed by the maximum PE order. 
This richer viewpoint enables a finer analysis of input--output data.

{\subsection{Input signal generator perspective on Willems' fundamental lemma}}
So, now suppose that \(\Signal{u}{0,T-1}\) is generated by the signal generator~\eqref{eq:sgenDT} satisfying Assumptions~\ref{ass:sglg}--\ref{ass:sgw0}.
By Lemma~\ref{lem:sgenPE}, PE condition in Proposition~\ref{prop:WFL+} can be restated in terms of the dimension \(N_g\) of the signal generator as follows.

\begin{corollary}
Suppose that we apply \(\Signal{u}{0,T-1}\), which is generated by the signal generator~\eqref{eq:sgenDT} satisfying Assumptions~\ref{ass:sglg}--\ref{ass:sgw0}, to the system~\eqref{eq:sys} as an input.
Depending on the signal generator dimension \(N_g\), the following hold for given \(L\in\N\).
    \begin{enumerate}[label=(\alph*)]
    \item If \(N_g<L\), then, for any \(T\), the input--output data fails to be informative for \(L\), for \emph{any} \(x(0)\in\R^n\).
    \item If \(L\leq N_g< L+n\), then, for any \(T\), the input--output data fails to be informative for \(L\), for \emph{some} \(x(0)\in\R^n\), generically.\footnote{
        Here, \emph{generically} is understood in the same sense as in Proposition~\ref{prop:WFL+}.
    }
    \item If \(N_g\geq L+n\) and \(T\geq 2N_g-1\), then the input--output data is informative for \(L\), for \emph{every} \(x(0)\in\R^n\).
    \end{enumerate}
\end{corollary}

In case~(b), it only establishes the existence of an initial condition \(x(0)\) for which the input--output data fails to be informative. 
Neither a size nor a structure of the exceptional set of initial conditions has been characterized.
It turns out, by further analysis in the next section, that it fails only for the measure zero portion of initial conditions.
In particular, the following statement is proved for almost all systems~\eqref{eq:sys}.
\begin{list}{}{\leftmargin=3mm \rightmargin=3mm}
    \item \itshape
    If \(L\leq N_g< L+n\) and \(T\geq N_g+n+L-1\), then the input--output data is informative for \(L\), for \emph{almost all} \(x(0)\in\R^n\).
\end{list}

Therefore, it is not necessary to design a signal generator with dimension \(N_g\geq L+n\).
Rather, a lower-order signal generator with \(N_g\geq L\) can sufficiently produce informative input--output data for almost all initial conditions.
In other words, the condition \(N_g\geq L\) is necessary and, in an almost-everywhere sense, sufficient for generating an informative input--output data.

{\section{Practical relaxation of Willems' fundamental lemma}}
When we interconnect the signal generator~\eqref{eq:sgenDT} to the system~\eqref{eq:sys}, we can understand the system behavior from a different perspective. This allows us to relax the input condition in the fundamental lemma. The subsequent analysis relies on the assumption below.

\begin{assumption}\label{ass:sgsigma}
\(\sigma(A)\cap\sigma(S_g)=\emptyset\).
\end{assumption}
This assumption holds generically: for any given \(A\), it is satisfied for almost all choices of \(S_g\), and, similarly, for any given \(S_g\), it is satisfied for almost all choices of \(A\).

{\subsection{Decomposition of the input--output Hankel matrix}}
According to the Assumption~\ref{ass:sgsigma}, the Sylvester equation
\begin{equation}\label{eq:syl}
    A\Pi + BL_g = \Pi S_g
\end{equation}
admits a unique solution \(\Pi\in\R^{n\times N_g}\)~\cite{a1}. 
With this \(\Pi\), we obtain
\begin{equation*}
    \begin{aligned}    
        x(t+1)-\Pi w(t+1) &= Ax(t)+(BL_g-\Pi S_g)w(t)\\
        &= A(x(t)-\Pi w(t)),
    \end{aligned}    
\end{equation*}
when we interconnect the signal generator~\eqref{eq:sgenDT} to the system~\eqref{eq:sys}. 
Thus
\begin{equation*}\label{eq:manifold}
    \mathscr{M} \coloneqq \left\{(x,w)\in\R^n\times \R^{N_g}\ : x=\Pi w\right\}
\end{equation*}
is an invariant manifold. From now on, we use \(\bar x\) to represent
\(\bar x(t) = x(t)-\Pi w(t)\) and write \(\bar x_0\) for the initial value \(\bar x(0)\).
Then, for signals \(w(t) = S_g^tw_0\) and \(\bar x(t) = A^t\bar x_0\), signals \(u(t)\) and \(y(t)\) are written as
\begin{equation*}
    \begin{aligned}
        u(t) &= L_gw(t),\\
        y(t) &= Cx(t) + Du(t) = (C\Pi+DL_g)w(t) + C\bar x(t)\\
        &\eqqcolon M_gw(t)+C\bar x(t).
    \end{aligned}    
\end{equation*}

\begin{remark}
    An invariant manifold \(\mathscr{M}\) has been discussed in moment matching based model reduction theory {\cite{i4}}. 
    \(M_g\) is called the forward moment of the system~\eqref{eq:sys} at \((L_g,S_g)\), in time-domain notion. It actually encodes the values of the transfer function of the system~\eqref{eq:sys} on \(\sigma(S_g)\),\footnote{
        The values of \(f\) on \(\sigma(A)\) mean \(f(\lambda_i),f^{(1)}(\lambda_i),\)\(\ldots,f^{(\eta_i-1)}(\lambda_i)\) for each \(\lambda_i\in\sigma(A)\), where \(f^{(k)}\) denotes the \(k\)-th derivative of \(f\)  and \(\eta_i\) is the maximum Jordan chain length associated with \(\lambda_i\).
    } which are called moments in frequency-domain notion.
\end{remark}

Then, we can write
\begin{equation*}\label{eq:signalio}
\begin{aligned}
    \begin{bmatrix}
        \Lift{u}{L}{t}\\
        \Lift{y}{L}{t}
    \end{bmatrix}
    \mkern-5mu&=\mkern-5mu \setlength{\arraycolsep}{2pt}\begin{bNiceArray}{c I c}[margin=2pt]
            \mkern-1mu\Obsvm{L}{L_g,S_g} & 0\\\H
            \mkern-1mu\Obsvm{L}{M_g,S_g} & \Obsvm{L}{C,A}\mkern-1mu
        \end{bNiceArray}\mkern-8mu\begin{bmatrix}
        w(t)\\\bar x(t)
    \end{bmatrix}
    \mkern-5mu\eqqcolon\mkern-2mu 
    \mathscr{L}_L\mkern-4mu\begin{bmatrix}
        w(t)\\\bar x(t)
    \end{bmatrix}\!,
\end{aligned}
\end{equation*}
and thus, can decompose the input--output Hankel matrix as
\begin{equation}\label{eq:decompio}
\begin{aligned}
    \mathscr{H}_{L,T}(u,y) 
    &= \mathscr{L}_L \setlength{\arraycolsep}{3pt}\begin{bNiceArray}{c c c}[margin=2pt]
        w(0)&\cdots&w(T-L)\\\H
        \bar x(0)&\cdots&\bar x(T-L)
    \end{bNiceArray}\\
    &= \mathscr{L}_L \begin{bNiceArray}{c}[margin=2pt]
        \Ctrbm{T-L+1}{S_g,w_0}\\\H
        \Ctrbm{T-L+1}{A,\bar x_0}
    \end{bNiceArray}\\
    &\eqqcolon \mathscr{L}_L\mathscr{R}_{T-L+1}.
\end{aligned}
\end{equation}

\noindent Now, \(\rank(\mathscr{H}_{L,T}(u,y))\) can be interpreted in terms of the structural properties of \(\mathscr{L}_L\) and \(\mathscr{R}_{T-L+1}\).

{\subsection{New fundamental lemma}}
Now we introduce our main result. The following theorem provides conditions on the signal generator---specifically, on its dimension \(N_g\)---that guarantee informative input--output data.

\begin{theorem}\label{thm:newFL}
Suppose that we interconnect the signal generator~\eqref{eq:sgenDT}, satisfying Assumptions~\ref{ass:sglg}--\ref{ass:sgsigma}, to the system~\eqref{eq:sys} and collect length-\(T\) input--output data.
Depending on the signal generator dimension \(N_g\), the following hold for given \(L\in\N\).
\begin{enumerate}[label=(\alph*)]
    \item If \(N_g<L\), then, for any \(T\), the input--output data fails to be informative for \(L\), for any \(x(0)\in\R^n\).
    \item If \(L\leq N_g < L+n\) and \(T\geq N_g + n + L - 1\), then the input--output data is informative for \(L\), for all \(x(0)\in \R^n\setminus E\), where \(E\) is a set of Lebesgue measure zero given in~\eqref{eq:except2}.
    \item If \(N_g\ge L+n\) and \(T\geq N_g + n + L - 1\), then the input--output data is informative for \(L\), for all \(x(0)\in\R^n\).
\end{enumerate}
\end{theorem}

The practical advantage of Theorem~\ref{thm:newFL} over Willems' fundamental lemma is most transparent 
when we focus on the case \(L=n+1\), assuming that the system order \(n\) is known. The resulting statement is given in the following corollary.

\begin{corollary}\label{cor:newFL}
Consider the signal generator~\eqref{eq:sgenDT} with dimension \(N_g=n+1\), satisfying Assumptions~\ref{ass:sglg}--\ref{ass:sgw0}.
If the signal generator response \(\Signal{u}{0,T-1}\),  which has length \(T=3n+1\), is applied to the system~\eqref{eq:sys} as an input, then
\begin{equation}\label{eq:rankcond}
    \rank(\mathscr{H}_{n+1,T}(u,y))=2n+1
\end{equation}
for almost all \(A\in\R^{n\times n}\) and \(x(0)\in\R^n\).
\end{corollary}

Since the condition~\eqref{eq:rankcond} implies that 
\(\ker\!\left(\mathscr{H}_{n+1,T}(u,y)^\top\right)\) directly provides coefficients of an input--output difference equation of the system~\eqref{eq:sys},
an input--output data satisfying such condition is important in data-driven control literature; for instance, see~{\cite{ia2}}.
For the condition~\eqref{eq:rankcond}, it has been common to use an input signal satisfying \(\PE(\Signal{u}{0,T-1})\geq2n+1\) with length \(T \geq 4n+1\) as required by Willems' fundamental lemma. 
However, Corollary~\ref{cor:newFL} shows that the same rank condition can be satisfied using a signal of minimum length \(T=3n+1\).
This shows us clearly how Theorem~\ref{thm:newFL} (or Corollary~\ref{cor:newFL}) practically improves upon Willems' fundamental lemma.

We now prove Theorem~\ref{thm:newFL}.
\begin{proof}[Proof of Theorem~\ref{thm:newFL}]
We begin by establishing an identity for \(\Obsvm{L}{M_g,S_g}\).
Applying Sylvester equation~\eqref{eq:syl} to each row gives
\begin{equation}\label{eq:obsvmm}\setlength{\arraycolsep}{2pt}
    \begin{aligned}
        &\Obsvm{L}{M_g,S_g} \\
        &=\begin{bmatrix}
            D\\
            CB&D\\
            \vdots&\ddots&\ddots\\
            CA^{L-2}B&\cdots&CB&D
        \end{bmatrix}\Obsvm{L}{L_g,S_g}
        +\Obsvm{L}{C,A}\Pi.
    \end{aligned}
\end{equation}

If \(N_g<L\), then \eqref{eq:obsvmm} and the observability of \((L_g,S_g)\) gives
\[
    \begin{aligned}
        \rank(\mathscr{L}_L)
        &= \rank\left(\begin{bNiceArray}{c I c}[margin=2pt]
            \Obsvm{L}{L_g,S_g}&0\\\H
            0&\Obsvm{L}{C,A}
        \end{bNiceArray}\right)\\
        &= N_g + \rank(\Obsvm{L}{C,A})
    \end{aligned}
\]
and \(\rank(\mathscr{H}_{L,T}(u,y)) \leq \rank(\mathscr{L}_L)<L+\rank(\Obsvm{L}{C,A})\), regardless of \(\mathscr{R}_{T-L+1}\).
Hence, (a) holds.

Otherwise, if \(N_g\geq L\), then determining the rank of \(\mathscr{H}_{L,T}(u,y)\) requires a more detailed structural analysis.
Consider a controllability decomposition of the pair \((A,\bar x_0)\), with \(C\) transformed accordingly, given by
\begin{equation*}
\label{eq:ctrbdecomp}
\begin{aligned}
    T^{-1}AT
    &=\begin{bmatrix}
        A_c & A_{12}\\
        0 & A_{\bar c}
    \end{bmatrix},
    &
    T^{-1}\bar x_0
    &=\begin{bmatrix}
        \bar x_{0,c}\\
        0
    \end{bmatrix},\\
    CT
    &=\begin{bmatrix}
        C_c & C_{\bar c}
    \end{bmatrix},
\end{aligned}
\end{equation*}
where \(T=\begin{bmatrix}T_c&T_{\bar c}\end{bmatrix}\) is invertible and
the columns of \(T_c\) form a basis for \(\im\left(\Ctrbm{}{A,\bar x_0}\right)\).
Using this, the decomposition of \(\mathscr{H}_{L,T}(u,y)\) in~\eqref{eq:decompio} takes the following form.
\begin{equation}\label{eq:decompctrb}
    \begin{aligned}
        &\mathscr{H}_{L,T}(u,y) = \mathscr{L}_L\mathscr{R}_{T-L+1}\\
        &=\setlength{\arraycolsep}{2pt}\begin{bNiceArray}{c I c}[margin=2pt]
            \mkern-1mu\Obsvm{L}{L_g,S_g} & 0\\\H
            \mkern-1mu\Obsvm{L}{M_g,S_g} & \Obsvm{L}{C_c,A_c}\mkern-1mu
        \end{bNiceArray}\mkern-5mu\begin{bNiceArray}{c}[margin=2pt]
            \Ctrbm{T-L+1}{S_g,w_0}\\\H
            \Ctrbm{T-L+1}{A_c,\bar x_{0,c}}
        \end{bNiceArray}.
    \end{aligned}
\end{equation}

According to Assumptions~\ref{ass:sgw0}--\ref{ass:sgsigma} and the controllability of \((A_c,\bar x_{0,c})\), 
it directly follows from the Popov-Belevitch-Hautus (PBH) test {\cite[Theorem~1]{a3}} that 
\[
\begin{bmatrix}
    \Ctrbm{T-L+1}{S_g, w_0}\\
    \Ctrbm{T-L+1}{A_c,\bar x_{0,c}}
\end{bmatrix} 
= \Ctrbm{T-L+1}{\begin{bmatrix}
    S_g&0\\0&A_c
\end{bmatrix}, \begin{bmatrix}
    w_0\\\bar x_{0,c}
\end{bmatrix}}
\]
has full row rank for \(T\geq N_g+n+L-1\).
It follows from this fact and an identity~\eqref{eq:obsvmm} that
\begin{equation*}
\begin{aligned}
    \rank(\mathscr{H}_{L,T}(u,y)) 
    = \rank\!\left(\!\setlength{\arraycolsep}{3pt}\begin{bNiceArray}{c I c}[margin=2pt]
            \mkern-1mu\Obsvm{L}{L_g,S_g}&0\\\H
            \mkern-1mu\Obsvm{L}{C,A}\Pi&\Obsvm{L}{C,A}T_c\mkern-1mu
        \end{bNiceArray}\!\right)\!.
\end{aligned}
\end{equation*}

The observability of \((L_g,S_g)\) enables the notion
\[
    \begin{bmatrix}
        \Gamma_1&\Gamma_2
    \end{bmatrix}\coloneqq \Pi\Obsvm{N_g}{L_g,S_g}^{-1}
\] where \(\Gamma_1\in\R^{n\times L}\) and \(\Gamma_2\in\R^{n\times(N_g-L)}\),
with \(\Gamma_2\) being understood as the null matrix when \(N_g=L\).
Then,
\[
\begin{aligned}
    &\setlength{\arraycolsep}{3pt}\begin{bNiceArray}{c I c}[margin=2pt]
            \Obsvm{L}{L_g,S_g}&0\\\H
            \Obsvm{L}{C,A}\Pi&\Obsvm{L}{C,A}T_c
        \end{bNiceArray}
        \setlength{\arraycolsep}{3pt}\begin{bNiceArray}{c I c}[margin=2pt]
            \Obsvm{N_g}{L_g,S_g}&0\\\H
            0&I_n
        \end{bNiceArray}^{-1}\\
    &=\setlength{\arraycolsep}{3pt}\begin{bNiceArray}{c I c I c}[margin=2pt]
            I_L&0&0\\\H
            \Obsvm{L}{C,A}\Gamma_1&\Obsvm{L}{C,A}\Gamma_2&\Obsvm{L}{C,A}T_c
        \end{bNiceArray}.
\end{aligned}
\]

Now, we obtain 
\[
\begin{aligned}
    &\rank(\mathscr{H}_{L,T}(u,y)) \\
    &= \rank\!\left(\!\setlength{\arraycolsep}{3pt}\begin{bNiceArray}{c I c I c}[margin=2pt]
            I_L&0&0\\\H
            \Obsvm{L}{C,A}\Gamma_1&\Obsvm{L}{C,A}\Gamma_2&\Obsvm{L}{C,A}T_c
        \end{bNiceArray}\!\right)\\
    &= \rank\!\left(\!\setlength{\arraycolsep}{3pt}\begin{bNiceArray}{c I c}[margin=2pt]
            I_L&0\\\H
            0&\Obsvm{L}{C,A}\begin{bmatrix}\Gamma_2& T_c\end{bmatrix}
        \end{bNiceArray}\!\right)\\
    &= L+ \rank(\Obsvm{L}{C,A}\!\begin{bmatrix}
        \Gamma_2& \Ctrbm{}{A,\bar x_0}
    \end{bmatrix}).
\end{aligned}
\]

Thus, \(\rank(\mathscr{H}_{L,T}(u,y))= L+ \rank(\Obsvm{L}{C,A})\) for all \(x(0)\in\R^n\setminus E\), where
\begin{equation}\label{eq:except2}
    E \coloneqq \left\{\bar x_0+\Pi w_0 \;\middle|
    \begin{array}{l}
        \rank\left(\Obsvm{L}{C,A}\begin{bmatrix}
        \Gamma_2& \Ctrbm{}{A,\bar x_0}
    \end{bmatrix}\right)\\
    <\rank(\Obsvm{L}{C,A})
    \end{array}
    \!\!\!\right\},
\end{equation}
supposing \(N_g\geq L\) and \(T\geq N_g+n+L-1\).

For the set \(E\), we have the following inclusions:
\begin{equation}\label{eq:exceptinclusion}
\begin{aligned}
    E
    &\subseteq \left\{\bar x_0+\Pi w_0 : \rank(\begin{bmatrix}
        \Gamma_2& \Ctrbm{}{A,\bar x_0}
    \end{bmatrix})<n\right\}\eqqcolon E_1\\
    &\subseteq \left\{\bar x_0+\Pi w_0 : \rank(\Ctrbm{}{A,\bar x_0})<n\right\}\eqqcolon E_2.
\end{aligned}
\end{equation}
By Lemma~\ref{lem:measure0} in the Appendix and the translation invariance of Lebesgue measure, \(E_2\) has Lebesgue measure zero, and so does \(E\).
Hence, (b) holds.
Moreover, for \(N_g\geq L+n\), 
\(\rank(\Gamma_2)=n\) by Lemma~\ref{lem:sylsol} in the Appendix, and hence, \(E=E_1=\emptyset\).
Therefore, (c) also holds.
\end{proof}

\begin{remark}
    To gain additional insight into the proof of Theorem~\ref{thm:newFL} (b), we can take a simple look at the matrix
    \[
    \mathscr{L}_{L,c}\coloneqq \begin{bNiceArray}{c I c}[margin=2pt]
            \Obsvm{L}{L_g,S_g} & 0\\\H
            \Obsvm{L}{M_g,S_g} & \Obsvm{L}{C_c,A_c}
    \end{bNiceArray}
    \]
    in~\eqref{eq:decompctrb} of the proof,
    which satisfies \(\rank(\mathscr{H}_{L,T}(u,y)) = \rank(\mathscr{L}_{L,c})\) for \(T\geq N_g+n_c+L-1\).

    It is clear from its size that \(\rank(\mathscr{L}_{L,c}) \le N_g + n_c\).
    Also, since \(\mathscr{L}_{L,c}\) is a block triangular matrix,
    \(\rank(\mathscr{L}_{L,c}) \ge \rank(\Obsvm{L}{L_g,S_g}) + \rank(\Obsvm{L}{C_c,A_c}) = \min\{N_g,L\} + \min\{n_c,L\}\).
    Thus, combining with the fact that \(\rank(\mathscr{H}_{L,T}(u,y)) \leq L+\rank(\Obsvm{L}{C,A})\), we have
    \[
    \begin{aligned}
        \min\{N_g,L\} + \min\{n_c, L\} \le \rank(\mathscr{H}_{L,T}(u,y)) \\
        \quad\le \min\{N_g+n_c, L+n, L+L\}
    \end{aligned}
    \]
    for \(T\geq N_g+n_c+L-1\).    
    So, if \(N_g\geq L\) and \(n=n_c\), then \(\rank(\mathscr{H}_{L,T}(u,y)) = L+\min\{n,L\} = L+ \rank(\Obsvm{L}{C,A})\).
    Here, \(n=n_c\) happens whenever \(\bar x_0 \notin E_2 \supset E\), 
    where \(E_2\) is a nonempty set of Lebesgue measure zero defined in~\eqref{eq:exceptinclusion}.
    
\end{remark}

Finally, to better understand Theorem~\ref{thm:newFL}, we can divide the input conditions into those on the signal generator dimension \(N_g\) and those on the data length \(T\).
We focus particularly on the case \(L\geq n+1\), 
in which the complete input--output behavior of system~\eqref{eq:sys} can be recovered from \(\mathscr{H}_{L,T}(u,y)\).

We first summarize the conditions on \(N_g\).
For \(L\geq n+1\), an exceptional set \(E\) in Theorem~\ref{thm:newFL} (b) equals a nonempty set \(E_1\) in~\eqref{eq:exceptinclusion}.
Thus, \(N_g\geq L+n\) is necessary and sufficient to obtain an informative input--output data for \emph{all} \(x(0)\in\R^n\).
However, to obtain such data for \emph{almost all} \(x(0)\in\R^n\), \(N_g\geq L\) is necessary and even sufficient.

On the other hand, with respect to \(T\), \(T\geq N_g+n+L-1\) is only a sufficient condition.
For \(L\geq n+1\), \(L+\rank(\Obsvm{L}{C,A}) = L+n\).
A necessary condition on \(T\) for \(\mathscr{H}_{L,T}(u,y)\in\R^{2L\times (T-L+1)}\) to have rank \(L+n\) is \(T \ge 2L+n-1\).
So there is room for a shorter input. 
Nevertheless, since \(N_g+n+L-1=2L+n-1\) for \(N_g=L\), the condition \(T\geq N_g+n+L-1\) appears to be fairly tight from the perspective of designing inputs.

{\section{Continuous-time fundamental lemma}}
Now, we establish the continuous-time counterpart of our signal generator based fundamental lemma.
Consider a continuous-time LTI SISO system
\begin{equation}\label{eq:sysCT}
    \begin{aligned}
        \dot x(t) &= Ax(t)+Bu(t)\in\R^n,\\
        y(t)      &= Cx(t)+Du(t),
    \end{aligned}
\end{equation}
which is controllable and observable.

{\subsection{Data representation of continuous-time systems}}
For a signal \(z:\R\to\R\) and \(L\in\N\), let
\[
    \LiftCT{z}{L}{t}\coloneqq [z(t)\; z^{(1)}(t)\; \cdots\; z^{(L-1)}(t)]^\top\in\R^{L},
\]
where \(z^{(\ell)}(t)\) is the \(\ell\)-th derivative of \(z\) with respect to \(t\) (whenever it exists).
We refer to \(\LiftCT{z}{L}{t}\) the \((L-1)\)-jet of \(z\) at time \(t\in\R\).\footnote{
    The terminology `jet' follows \cite{m2}.
} 
Then, we define the set of input--output \((L-1)\)-jets of the system~\eqref{eq:sysCT} as
\[
    \mathscr{B}_L^{\,c}
    \coloneqq
    \left\{\!
    \begin{bmatrix}
        \LiftCT{u}{L}{t_0}\\[1mm]
        \LiftCT{y}{L}{t_0}
    \end{bmatrix}\!\in\R^{2L}
    \;\middle|
    \begin{array}{l}
        \exists\, x:\R\to\R^n \text{ such that}\\ 
        (u,x,y)\text{ satisfies } \eqref{eq:sysCT}\\
        \text{in a neighborhood of }t_0
    \end{array}
    \!\!\!\right\}.
\]
Here, we restrict our attention to smooth input signals so that the jets are well-defined.
Since~\eqref{eq:sysCT} is an LTI system, \(\mathscr{B}_L^{\,c}\) is a linear subspace of \(\R^{2L}\).

We consider the continuous-time system recovery, that is, recovering \(\mathscr{B}_L^{\,c}\) from input--output data. 
As in the discrete-time case, it suffices to find a spanning set of \(\mathscr{B}_L^{\,c}\) consisting of collected input--output \((L-1)\)-jets.

We regard the input--output data over a time interval \([0,T]\) as the observed trajectories \(u:[0,T]\to\R\) and \(y:[0,T]\to\R\).
For sampling instants \(t_1,\ldots,t_{k}\in(0,T)\), we can define the continuous-time Hankel-type matrix constructed from the corresponding input--output \((L-1)\)-jets as
\[
\mathscr{H}_{L,k}^{\,c}(u,y)
\coloneqq
\begin{bmatrix}
\LiftCT{u}{L}{t_1} & \cdots & \LiftCT{u}{L}{t_k} \\
\LiftCT{y}{L}{t_1} & \cdots & \LiftCT{y}{L}{t_k}
\end{bmatrix}
\in\R^{2L\times k}.
\]
By construction, \(\im(\mathscr{H}_{L,k}^{\,c}(u,y))\) is a linear subspace of \(\mathscr{B}_L^{\,c}\). 
The remaining question is which input conditions ensure that
\begin{equation}\label{eq:imHankelCT}
    \im(\mathscr{H}_{L,k}^{\,c}(u,y))=\mathscr{B}_L^{\,c}.
\end{equation}
An equivalent rank characterization for condition~\eqref{eq:imHankelCT} is induced similarly to the discrete-time case in {\cite{m1}} as follows.

\begin{proposition}[\cite{m1}]\label{prop:colspaceCT}
Consider the set of sampled input--output \((L-1)\)-jets
\[
    \left\{
    \begin{bmatrix}
        \LiftCT{u}{L}{t}\\
        \LiftCT{y}{L}{t}
    \end{bmatrix}:\ t\in\{t_1,t_2,\ldots,t_k\}
    \right\}
\]
generated by~\eqref{eq:sysCT}. 
The following are equivalent:
\begin{enumerate}[label=(\alph*)]
    \item The condition~\eqref{eq:imHankelCT} is satisfied.
    \item \(\rank(\mathscr{H}_{L,k}^{\,c}(u,y))
    \mkern-1mu=\mkern-1mu 
    \dim(\mathscr{B}_L^{\,c})
    \mkern-1mu=\mkern-1mu
    L + \rank(\Obsvm{L}{C,\mkern-1muA})\).
\end{enumerate}
\end{proposition}
We say that an input--output data \((u:[0,T]\to\R, y:[0,T]\to\R)\) is informative for \(L\) if there exist \({t_1,\ldots,t_k}\in(0,T)\) such that input--output \((L-1)\)-jets at \({t_1,\ldots,t_k}\) satisfy either of the equivalent conditions in Proposition~\ref{prop:colspaceCT}.

{\subsection{New fundamental lemma for continuous-time systems}}
We introduce an autonomous LTI single-output system
\begin{equation}\label{eq:sgenCT}
    \begin{aligned}
        \dot w(t)&=S_gw(t)\in\R^{N_g},\quad w(0)=w_0\\
        u(t)&=L_gw(t)\in\R,
    \end{aligned}
\end{equation}
as a signal generator. 
Under the continuous-time counterparts of Assumptions~\ref{ass:sglg}--\ref{ass:sgsigma}, the analysis proceeds similarly to the discrete-time case.

We interconnect the signal generator~\eqref{eq:sgenCT} to the system~\eqref{eq:sysCT} 
and let \(\bar x(t) = x(t)-\Pi w(t), M_g=C\Pi+DL_g\) for the solution \(\Pi\) of the Sylvester equation \(A\Pi+BL_g=\Pi S_g\).
We write \(\bar x_0\) for the initial value \(\bar x(0)\).
Since \(\dot x(t)-\Pi \dot w(t) = A(x(t)-\Pi w(t))\), we have
\begin{equation*}
    \begin{aligned}
        u^{(\ell)}(t) &= L_gS_g^\ell w(t),\\
        y^{(\ell)}(t) &= M_gS_g^\ell w(t) + CA^\ell \bar x(t)
    \end{aligned}    
\end{equation*}
for \(w(t) = e^{S_g t}w_0\) and \(\bar x(t) = e^{At}\bar x_0\).
Thus, we recover
\begin{equation*}\label{eq:signalioCT}
\begin{aligned}
    \begin{bmatrix}
        \LiftCT{u}{L}{t}\\
        \LiftCT{y}{L}{t}
    \end{bmatrix}
    \mkern-5mu&=\mkern-5mu \setlength{\arraycolsep}{2pt}\begin{bNiceArray}{c I c}[margin=2pt]
            \mkern-2mu\Obsvm{L}{L_g,S_g}&0\\\H
            \mkern-2mu\Obsvm{L}{M_g,S_g}&\Obsvm{L}{C,A}\mkern-2mu
        \end{bNiceArray}\mkern-8mu\begin{bmatrix}
        w(t)\\\bar x(t)
    \end{bmatrix}
    \mkern-5mu\eqqcolon\!
    \mathscr{L}_L\mkern-4mu\begin{bmatrix}
        w(t)\\\bar x(t)
    \end{bmatrix}\!.
\end{aligned}
\end{equation*}

For suitably chosen \(k\geq N_g+n\) sampling instants, \(t_1,\ldots,t_k\in(0,T)\),
\[
\begin{aligned}
\im\!\left(\setlength{\arraycolsep}{3pt}\begin{bNiceArray}{c c c}[margin=2pt]
    w(t_1)&\cdots&w(t_k)\\\H
    \bar x(t_1)&\cdots&\bar x(t_k)
\end{bNiceArray}\right)
&= \spanop\left\{\begin{bmatrix}
    w(t)\\\bar x(t)
\end{bmatrix}:\ t\in(0,T)\right\}\\
&= \im\!\left(
    \Ctrbm{}{\begin{bmatrix}
        S_g&0\\0&A
    \end{bmatrix},\begin{bmatrix}
        w_0\\\bar x_0
    \end{bmatrix}}
    \right).
\end{aligned}
\]
Thus, for such suitably chosen sampling instants,
\[
\begin{aligned}
    \im\left(\mathscr{H}_{L,k}^{\,c}(u,y)\right)
    &= \im\!\left(\mathscr{L}_L
    \Ctrbm{}{\begin{bmatrix}
        S_g&0\\0&A
    \end{bmatrix},\begin{bmatrix}
        w_0\\\bar x_0
    \end{bmatrix}}
    \right)\\
    &\eqqcolon \im\left(\mathscr{L}_L\mathscr{R}^{\,c}\right),
\end{aligned}
\]
and hence, \(\rank(\mathscr{H}_{L,k}^{\,c}(u,y)) = \rank\left(\mathscr{L}_L\mathscr{R}^{\,c}\right)\).

At this point, we state continuous-time counterpart of our signal generator based fundamental lemma. The proof is analogous to Theorem~\ref{thm:newFL}.

\begin{theorem}\label{thm:newFLCT}
Consider the system~\eqref{eq:sysCT} and the signal generator~\eqref{eq:sgenCT} such that \((L_g,S_g)\) is observable, \((S_g,w_0)\) is controllable, and \(\sigma(A)\cap\sigma(S_g)=\emptyset\).
Suppose that we interconnect the signal generator~\eqref{eq:sgenCT} to the system~\eqref{eq:sysCT} and observe \((u,y)\) on \([0,T]\) for \(T> 0\).
Depending on the signal generator dimension \(N_g\), the following hold for given \(L\in\N\).
\begin{enumerate}[label=(\alph*)]
    \item If \(N_g<L\), then the input--output data fails to be informative for \(L\), for any \(x(0)\in\R^n\).
    \item If \(L\leq N_g < L+n\), then the input--output data is informative for \(L\), for \emph{almost all} \(x(0)\in\R^n\), except for a set of Lebesgue measure zero\footnote{This exceptional set is the continuous-time analogue of \(E\) in \eqref{eq:except2}.}.
    \item If \(N_g\geq L+n\), then the input--output data is informative for \(L\), for all \(x(0)\in\R^n\).
\end{enumerate}
\end{theorem}

{\section{Conclusion}}

In this paper, we provide a practical relaxation of Willems’ fundamental lemma.
This was possible by shifting our attention from persistently exciting input trajectories to input signal generators, that is, by returning from behavioral theory to dynamical systems theory.
The proposed result provides a necessary and sufficient condition on the signal generator's dimension for generating informative input–output data for almost all systems and initial conditions. 
Our input design guideline is practically efficient, not only because it provides a simple structure for designing input signals, 
but also because it offers additional input choices beyond the class originally considered by Willems’ fundamental lemma, while remaining useful for almost all systems and initial conditions.

{\section{Acknowledgments}}
The ChatGPT has been used to improve the syntax and grammar of several paragraphs in the manuscript.


\appendix

\section{} 

\begin{lemma}\label{lem:randomPE}
For \(z=[z_1\;z_2\;\cdots\;z_{2K-1}]^\top\in\R^{2K-1}\), define a multi-variate function \(H\) by
\[
    \setlength{\arraycolsep}{2pt}
    H(z)=\!\begin{bmatrix}
        z_1&z_2&\cdots &z_{K}\\
        z_2&z_3&\cdots &z_{K+1}\\
        \vdots&\vdots&\ddots&\vdots\\
        z_{K}&z_{K+1}&\cdots&z_{2K-1}
    \end{bmatrix}.
\]
Then, the set \( \left\{z\in\mathbb{R}^{2K-1}: \rank(H(z))<K \right\} \)
has Lebesgue measure zero.
\end{lemma}
\begin{proof}
A function \(f(z) = \det(H(z))\) is a nonzero polynomial function, 
since \(f([0_{1\times (K-1)}\;1\;0_{1\times (K-1)}]^\top)\neq 0\).
Hence, the set \(\left\{z\in\mathbb{R}^{2K-1}: \rank(H(z))<K \right\} = \{z\in\R^{2K-1} : f(z) = 0\}\) has Lebesgue measure zero {\cite{a4}}. 
\end{proof}

\begin{lemma}\label{lem:randomPEunif}
Suppose that the terms of \(u_{[0,2K-2]}\) are 
drawn independently from the uniform distribution on \([a,b]\subset \R\), where \(a<b\).
Then \(H_K(u_{[0,2K-2]})\) has full rank almost surely.
\end{lemma}
\begin{proof}
Let \(U \coloneqq [u(0)\;u(1)\;\cdots\;u(2K-2)]^\top \in \R^{2K-1}\).
Using the multi-variate function \(H\) defined in Lemma~\ref{lem:randomPE},
let \(\mathcal{Z} \coloneqq \left\{z\in\mathbb{R}^{2K-1}: \rank(H(z))<K \right\}\), which has Lebesgue measure zero by Lemma~\ref{lem:randomPE}.
Then, the probability that \(U\) belongs to \(\mathcal{Z}\) is
\[
    \Pr\!\left\{U\in\mathcal{Z}\right\}
    =
    \frac{1}{(b-a)^{2K-1}}
    \int_{\mathcal{Z}\cap[a,b]^{2K-1}} 1 dz
    =
    0.
\]
Thus, \(H(U) = H_K(u_{[0,2K-2]})\) has full rank almost surely.
\end{proof}

\begin{lemma}\label{lem:measure0}
Consider \(A\in\R^{n\times n}\) and \(B\in\R^{n\times 1}\), where \((A,B)\) is controllable.
Then, for almost all \(\bar x_0\in\R^n\),
\[\rank(\Ctrbm{}{A,\bar x_0})=n.\] 
\end{lemma}
\begin{proof}
For \(z\in\R^n\), define a multi-variate function \(f\) by
\[
    f(z) = \det(\Ctrbm{}{A,z}).
\]
Since \(f(B) \neq 0\), \(f\) is a nonzero polynomial function.
Hence, the set \(\{z\in\R^{n} : f(z) = 0\}\) has Lebesgue measure zero {\cite{a4}} and 
\(\Ctrbm{}{A,z}\) has full rank for almost all \(z\in\R^{n}\).
\end{proof}

\begin{lemma}\label{lem:sylsol}
Suppose that Assumptions \ref{ass:sglg} and~\ref{ass:sgsigma} hold.
Let \(\Pi\in\R^{n\times N_g}\) be the solution of the Sylvester equation~\eqref{eq:syl} 
and let \(\Gamma = \Pi\Obsvm{N_g}{L_g,S_g}^{-1}\).
We partition \(\Gamma\) columnwise as \(\Gamma = \begin{bmatrix} \Gamma_1 & \Gamma_2 \end{bmatrix}\),
where \(\Gamma_1 \in \R^{n\times (N_g-\ell)}\) and \(\Gamma_2 \in \R^{n\times \ell}\) for any \(\ell\in\{1,\ldots,N_g\}\),
with $\Gamma_1$ being understood as the null matrix when $\ell = N_g$.
Then, \(\rank(\Gamma_2)=\rank(\Ctrbm{\ell}{A,B})\).
\end{lemma}
\begin{proof}
Let \(p(z) \coloneqq \det(zI-S_g)=z^{N_g}+\xi_{N_g-1}z^{N_g-1}+\cdots+\xi_{1}z+\xi_0\). 
Let 
\[
\begin{aligned}
    \Sigma &= \begin{bmatrix}
        0 & 1 & 0 & \cdots & 0 \\[-0.5mm]
        0 & 0 & 1 & \cdots & 0 \\[-1.5mm]
        \vdots & \vdots& \vdots& \ddots & \vdots \\[-0.5mm]
        0 & 0 & 0 & \cdots & 1 \\[-0.5mm]
        - \xi_0 & - \xi_1 & - \xi_2 & \cdots & - \xi_{N_g-1}
    \end{bmatrix}\in\R^{N_g\times N_g},\\
    \Lambda &= \begin{bmatrix}
        1 & 0& 0& \cdots & 0
    \end{bmatrix}\in\R^{1\times N_g}.
\end{aligned}
\]
Then, \(\Lambda \Obsvm{}{L_g,S_g} = L_g\), \(\Sigma \Obsvm{}{L_g,S_g} = \Obsvm{}{L_g,S_g}S_g\), and \(\sigma(\Sigma) = \sigma(S_g)\). 
It follows that
\(A\Gamma + B\Lambda = A\Pi\Obsvm{}{L_g,S_g}^{-1}+BL_g\Obsvm{}{L_g,S_g}^{-1} = \Pi S_g\Obsvm{}{L_g,S_g}^{-1} = \Pi \Obsvm{}{L_g,S_g}^{-1} \Sigma = \Gamma \Sigma\).

Let \(\Gamma = \begin{bmatrix}
    \gamma_1&\gamma_2&\cdots& \gamma_{N_g-1}& \gamma_{N_g}
\end{bmatrix}\). 
By examining each column of \(\Gamma\Sigma\) in turn, we obtain
\[
\begin{aligned}
    \begin{cases}
        A\gamma_1 +B = -\xi_0 \gamma_{N_g},\\
        A\gamma_k = \gamma_{k-1} - \xi_{k-1} \gamma_{N_g},\quad k =2,3,\ldots,N_g.
    \end{cases}
\end{aligned}
\]
Starting from \(\gamma_{N_g-1} = (A + \xi_{N_g-1}I)\gamma_{N_g}\), we obtain
\(\gamma_k = (A^{N_g-k} + \xi_{N_g-1}A^{N_g-k-1} + \cdots + \xi_k I)\gamma_{N_g}\) for \(k = N_g-2 ,\ldots , 2,1\), successively.
Also, \(B = -A\gamma_1  -\xi_0 \gamma_{N_g} = -p(A) \gamma_{N_g}\).
It follows that
\[
\begin{aligned}
    \Gamma_2 
    &= \begin{bmatrix}
        \gamma_{N_g -\ell+1}&\cdots& \gamma_{N_g-1}& \gamma_{N_g}
    \end{bmatrix}\\
    &= \begin{bmatrix}
        \gamma_{N_g} & A\gamma_{N_g} & \cdots & A^{\ell-2}\gamma_{N_g} & A^{\ell-1}\gamma_{N_g}
    \end{bmatrix} \Xi
\end{aligned}
\]
where \(\Xi\) is an invertible matrix such that 
\[
\Xi = \begin{bmatrix}
        \xi_{N_g-\ell+1} & \xi_{N_g-\ell+2} & \cdots & \xi_{N_g-1} & 1\\
        \xi_{N_g-\ell+2} & \xi_{N_g-\ell+3} & \cdots & 1 \\
        \vdots & \vdots & \iddots\\
        \xi_{N_g-1}&1 \\
        1
    \end{bmatrix}.
\]

For a set \(\{\lambda_i\}_{i=1}^n\) of eigenvalues of \(A\), \(\{p(\lambda_i)\}_{i=1}^n\) is a set of eigenvalues of \(p(A)\).
Since \(\sigma(A)\cap \sigma(S_g)=\emptyset\) (Assumption~\ref{ass:sgsigma}), 
\(\det(p(A)) = p(\lambda_1)\cdots p(\lambda_n) \neq 0\), and hence, \(p(A)\) is invertible.
Since \(p(A)\) is invertible and commutes with \(A\), we have \(A = p(A)^{-1} A p(A)\) and \(\gamma_{N_g} = -p(A)^{-1}B\).
Thus, \(\Ctrbm{\ell}{A,\gamma_{N_g}} = -p(A)^{-1} \Ctrbm{\ell}{A,B}\).
Therefore, \(\rank(\Gamma_2) = \rank(\Ctrbm{\ell}{A,\gamma_{N_g}}) = \rank(\Ctrbm{\ell}{A,B})\).
\end{proof}

\end{document}